%% file: hopau-eq.tex
\documentclass[a4paper,UKenglish]{lipics-v2016}
\usepackage{microtype}
\usepackage{amsmath}
\usepackage{amsfonts}
\usepackage{amssymb}
\usepackage{bussproofs}
\usepackage{tikz}
\theoremstyle{theorem}
\newtheorem{proposition}{Proposition}
\include{SpecialSymbols}
\usepackage{cancel}

\makeatletter
\def\infrule#1#2#3{\@ifnextchar[{\@infrule{#1}{#2}{#3}}{\@infrule{#1}{#2}{#3}[*]}}%

\def\@infrule#1#2#3[#4]{
 \par\bigbreak
 \vtop{
  \leavevmode\null
  \textsf{#1:}\kern.5em\textbf{#2}\\[\smallskipamount]
  \ifx*#4 
   \null\qquad$#3$
  \else 
   \setbox30=\hbox{\qquad$#3$\qquad#4}%
   \ifdim\wd30>\hsize
    \null\qquad$#3$\par\kern-\parskip\smallskip#4
   \else
    \null\qquad$#3$\qquad#4
   \fi
  \fi
 }%
} \makeatother
\newcommand{\block}[2]{\ensuremath{\left(#1 \ ; \ \left\lbrace  #2 \right\rbrace \right) }}
\newcommand{\kdet}[3]{\ensuremath{\mathit{det}\left(#1, #2 , #3 \right) }}
\newcommand{\skdet}[3]{\ensuremath{\mathit{det}_{s}\left(#1, #2 , #3 \right) }}

\title{Higher-Order Equational Pattern Anti-Unification [Preprint]\footnote{This research is supported by the FWF project P28789-N32.}}
\titlerunning{Higher-order Equational Anti-unification}

\author[1]{David M. Cerna}
\author[2]{Temur Kutsia }
\affil[1]{Research Institute for Symbolic Computation, Johannes Kepler University, Linz, Austria\\
  \texttt{david.cerna@risc.jku.at}}
\affil[2]{Research Institute for Symbolic Computation, Johannes Kepler University, Linz, Austria\\
  \texttt{temur.kutsia@risc.jku.at}}
\authorrunning{D. M. Cerna and T. Kutsia} 

\Copyright{David M. Cerna and Temur Kutsia}


\makeatletter
\def\infrule#1#2#3{\@ifnextchar[{\@infrule{#1}{#2}{#3}}{\@infrule{#1}{#2}{#3}[*]}}%

\allowdisplaybreaks

\def\@infrule#1#2#3[#4]{
 \par\bigbreak
 \vtop{
  \leavevmode\null
  \textsf{#1:}\kern.5em\textbf{#2}\\[\smallskipamount]
  \ifx*#4 
   \null\qquad$#3$
  \else 
   \setbox30=\hbox{\qquad$#3$\qquad#4}%
   \ifdim\wd30>\hsize
    \null\qquad$#3$\par\kern-\parskip\smallskip#4
   \else
    \null\qquad$#3$\qquad#4
   \fi
  \fi
 }%
} \makeatother

\begin{document}
\maketitle

\begin{abstract}
We consider anti-unification for simply typed lambda terms in associative, commutative, and associative-commutative theories and develop a sound and complete algorithm which takes two lambda terms and computes their generalizations in the form of higher-order patterns. The problem is finitary: the minimal complete set of generalizations contains finitely many elements. We define the notion of optimal solution and investigate special fragments of the problem for which the optimal solution can be computed in linear or polynomial time.
\end{abstract}

\section{Introduction}

Anti-unification algorithms aim at computing generalizations for given terms. A generalization of $t$ and $s$ is a term $r$ such that $t$ and $s$ are substitution instances of $r$. Interesting generalizations are those that are least general (lggs). However, it is not always possible to have a unique least general generalization. In these cases the task is either to compute a minimal complete set of generalizations, or to impose restrictions so that uniqueness is guaranteed.

Anti-unification, as considered in this paper, uses both of these ideas. The theory is simply-typed lambda calculus, where some function symbols may be associative, commutative, or associative-commutative. A,-, C-, and AC-anti-unification is finitary even for first-order terms, and a modular algorithm has been proposed in~\cite{DBLP:journals/iandc/AlpuenteEEM14} to compute the corresponding minimal complete set of generalizations. Anti-unification for simply typed lambda terms can be restricted to compute generalizations in the form of Miller's patterns~\cite{DBLP:journals/logcom/Miller91}, which makes it unitary, and the single least general generalization can be computed in linear time by the algorithm proposed in~\cite{DBLP:journals/jar/BaumgartnerKLV17}. These two approaches combine nicely with each other when one wants to develop a higher-order equational anti-unification algorithm, and we illustrate it in this paper. Basically, it extends the syntactic\footnote{We refer to the higher-order anti-unification algorithm from \cite{DBLP:journals/jar/BaumgartnerKLV17} as syntactic, although it works modulo $\beta\eta$-conversion.} generalization rules from \cite{DBLP:journals/jar/BaumgartnerKLV17} by equational decomposition rules inspired by the those from \cite{DBLP:journals/iandc/AlpuenteEEM14}, getting a modular algorithm in which different equational axioms for different function symbols can be combined automatically. The algorithm takes a pair of simply typed lambda terms and returns a set of their generalizations in the form of higher-order patterns. It is terminating, sound, and complete. However, the number of nondeterministic choices when decompositg may result in a large search tree. Although each branch can be developed in linear time, there can be too many of them to search efficiently.

This is the problem that we address in the second part of the paper. The idea is to use a greedy approach: introduce an optimality criterion, use it to select an anti-unification problem among different alternatives obtained by a decomposition rule, and try to solve only that. In this way, we would only compute one generalization. Checking the criterion and selecting the right branch should be done ``reasonably fast''. To implement this idea, we introduce conditions on the form of anti-unification problems which are guarantee to compute ``optimal'' solutions, and study the corresponding complexities. In particular, we identify conditions for which A-, C-, and AC-generalizations can be computed in linear time. We also study how the complexity changes by relaxing these conditions.


Higher-order anti-unification has been investigated by various authors from different application 
perspective. Research  has been focused mainly on the investigation of special classes for which 
the uniqueness of lgg is guaranteed. Some application areas include proof generalization~
\cite{DBLP:conf/lics/Pfenning91}, higher-order term indexing~\cite{DBLP:journals/tocl/Pientka09}, 
cognitive modeling and analogical reasoning~\cite{DBLP:journals/connection/BesoldKP17,DBLP:series/sci/SchmidtKGK14}, recursion scheme detection in functional programs~\cite{DBLP:journals/fgcs/BarwellBH18}, inductive synthesis of recursive functions~\cite{DBLP:books/sp/Schmid03}, just 
to name a few. Two higher-order anti-unification algorithms~\cite{DBLP:journals/iandc/BaumgartnerK17,DBLP:journals/jar/BaumgartnerKLV17} are included in an online open-source anti-
unification library~\cite{baumgartner15,DBLP:conf/jelia/BaumgartnerK14}. This related work does 
not consider anti-unification with higher-order terms in the presence of equational axioms. 
However, such a combination can be useful, for instance, for developing indexing techniques for 
higher-order theorem provers~\cite{DBLP:conf/cade/LibalS16} or in higher order program manipulation 
tools. 

The organization of the paper is as follows: In Section~\ref{sect:ho:patterns} we introduce the main notions and define the problem. In Section~\ref{sect:hopau} we recall the higher-order anti-unification algorithm from \cite{DBLP:journals/jar/BaumgartnerKLV17}. In Section~\ref{sect:genAC} we extend the algorithm with equational decomposition rules. Section~\ref{sec:special:fragments} is devoted to the introduction of computationally well-behaved fragments of anti-unification problems. The next sections describe the behavior of equational anti-unification algorithms on these fragments: In Section~\ref{sect:special:A} we discuss associative generalization and speak about optimality. Sections~\ref{sect:special:C} and~\ref{sect:special:AC} are about C- and AC-generalizations. 

\section{Preliminaries}
\label{sect:ho:patterns}

This work builds upon the to be introduced background theory and the  results of \cite{DBLP:conf/rta/BaumgartnerKLV13,DBLP:journals/jar/BaumgartnerKLV17}. Higher-order signatures are composed of \emph{types} constructed from a set
of \emph{base types} (typically $\delta$) using the grammar $\tau ::= \delta \mid \tau \rightarrow \tau$. We will consider 
$\rightarrow$ to  associate right unless otherwise stated. \emph{Variables} (typically
$X, Y, Z, x, y, z, a, b, \ldots$) as well as \emph{constants} (typically $f,
c, \ldots$) are assigned types from the set of types constructed using the above grammar. \emph{$\lambda$-terms} (typically $t,s,u,\ldots$) are constructed using the grammar  $t ::= x \mid c \mid \lambda x.t \mid t_1\ t_2$
where $x$ is a variable and $c$ is a constant, and are typed using the type construction mentioned above. Terms of the form $(\ldots (h\ t_1 ) \ldots t_m )$, where $h$
is a constant or a variable, will be written as $h(t_1 , \ldots , t_m
)$, and terms of the form $\lambda x_1 . \ldots .\lambda x_n .t$ as
$\lambda x_1 , \ldots , x_n .t$. We use $\vect{x}$ as a short-hand for
$x_1 , \ldots , x_n$. This basic language will be extended by higher-order constants satisfying equational axioms. When necessary, we write a $\lambda$-term $t$ together with its type $\alpha$ as $t:\alpha$.

Every higher-order constant $c$ will have an associated set of axioms, denoted by $\mathit{Ax}(c)$. If $\mathit{Ax}(c)$ is empty then $c$ does not have any associated properties and is called \emph{free}. Otherwise, $\mathit{Ax}(f)\subseteq \lbrace {\sf A},{\sf C}\rbrace$ where $\sf A$ is associativity, i.e. $f(a,f(b,c)) \equiv f(f(a,b),c)$ and $\sf C$ is commutativity, i.e. $f(a,b) \equiv f(b,a)$. Note that only functions of the type $\alpha\rightarrow \alpha\rightarrow \alpha$ are allowed to have equational properties. We assume that terms are written in \emph{flattened form}, obtained by replacing all subterms of the form$f(t_1,\ldots,f(s_1,\ldots,s_m),\ldots t_n)$ by $f(t_1,\ldots,s_1,\ldots,s_m,\ldots t_n)$, where ${\sf A}\in Ax(f)$. Also, by convention, the term $f(t)$ stands for $t$, if ${\sf A}\in Ax(f)$.
Other standard notions of the simply typed $\lambda$-calculus, like
bound and free occurrences of variables, $\alpha$-conversion,
$\beta$-reduction, $\eta$-long $\beta$-normal form, etc. are defined
as usual (see~\cite{barendregt84,DBLP:books/el/RV01/Dowek01}). By default, terms
are assumed to be written in $\eta$-long $\beta$-normal
form. Therefore, all terms have the form $\lambda x_1, \ldots ,
x_n .h(t_1 , \ldots , t_m )$, where $n, m \geq 0$, $h$ is either a
constant or a variable, $t_1 , \ldots , t_m$ have this form, and
the term $h(t_1 , \ldots , t_m )$ has a basic type.

The set of free variables of a term $t$ is denoted by $\vars(t)$. When
we write an equality between two $\lambda$-terms, we mean that they
are equivalent modulo $\alpha$, $\beta$ and $\eta$ equivalence.

The \emph{size} of a term $t$, denoted $\size{t}$, is defined recursively as 
$\size{h(t_1,\ldots,t_n)}= 1+ \sum_{i=1}^n\size{t_i}$ and
$\size{\lambda x.t} = 1 + \size{t}$. The \emph{depth} of a term $t$, denoted $\depth(t)$ is defined
recursively as 
$\depth(h(t_1,\ldots,t_n))= 1+ \max_{i\in\{1,\dots,n\}} \depth(t_i)$ and
$\depth(\lambda x.t) = 1 + \depth(t)$. For a term $t= \lambda x_1,\ldots,x_n. h(t_1,\ldots,t_m)$ with $n,m\ge
0$, its \emph{head} is defined as $\head(t)=\nobreak h$.

A \emph{higher-order pattern} is a $\lambda$-term where, when written
in $\eta$-long $\beta$-normal form, all free variable occurrences are
applied to lists of pairwise distinct ($\eta$-long forms of) bound
variables. For instance, $\lambda x. f (X (x), Y )$, $f (c, \lambda
x.x)$ and $\lambda x.\lambda y.X (\lambda z.x(z), y)$ are patterns,
while $\lambda x. f (X (X (x)), Y )$, $f (X (c), c)$ and $\lambda
x.\lambda y.X (x, x)$ are not.

\emph{Substitutions} are finite sets of pairs $\{X_1\mapsto
t_1,\ldots, X_n\mapsto t_n\}$ where $X_i$ and $t_i$ have the same type
and the $X$'s are pairwise distinct variables. They can be extended to
type preserving functions from terms to terms as usual, avoiding
variable capture. The notions of substitution \emph{domain} and
\emph{range} are also standard and are denoted, respectively, by
$\dom$ and $\ran$.

We use postfix notation for substitution applications, writing
$t\sigma$ instead of $\sigma(t)$. As usual, the application $t\sigma$
affects only the free occurrences of variables from $\dom(\sigma)$ in
$t$. We write $\vect{x}\sigma$ for $x_1\sigma,\ldots,x_n\sigma$, if
$\vect{x}=x_1,\ldots,x_n$. Similarly, for a set of terms $S$, we
define $S\sigma = \{t\sigma \mid t\in S\}$.  The \emph{composition} of
$\sigma$ and $\vartheta$
 is written as juxtaposition
$\sigma\vartheta$ and is defined as $x(\sigma\vartheta) = (x\sigma)\vartheta$ for all $x$. 
Another standard operation, \emph{restriction}
of a substitution $\sigma$ to a set of variables $S$, is denoted by
$\sigma|_S$.

A substitution $\sigma_1$ is \emph{more general} than $\sigma_2$,
written $\sigma_1\preceq \sigma_2$, if there exists $\vartheta$ such
that $X\sigma_1\vartheta = X\sigma_2$ for all $X \in
\dom(\sigma_1)\cup \dom(\sigma_2)$. The strict part of this relation
is denoted by $\prec$. The relation $\preceq$ is a partial order and
generates the equivalence relation which we denote by $\simeq$.  We
overload $\preceq$ by defining $s\preceq t$ if there exists a
substitution $\sigma$ such that $s\sigma = t$. The focus of this work is generalization 
in the presence of equational axioms thus we need a more general concept of ordering of 
substitutions/terms by their generality. We say that two terms $s,t$ are  $s =_{\mathcal{E}} t$ if they are equivalent modulo $\mathcal{E}\subseteq \{{\sf A},{\sf C}\}$. For example, $f(a,f(b,c)) \neq f(f(a,b),c)$ but, $f(a,f(b,c)) =_{\lbrace A \rbrace} f(f(a,b),c)$. Under this notion of equality we can say that a substitution $\sigma_1$ is \emph{more general modulo an equational theory  $\mathcal{E}\subseteq \{{\sf A},{\sf C}\}$} than $\sigma_2$ written $\sigma_1\preceq_{\mathcal{E}} \sigma_2$ if there exists $\vartheta$ such
that $X\sigma_1\vartheta =_{\mathcal{E}} X\sigma_2$ for all $X \in
\dom(\sigma_1)\cup \dom(\sigma_2)$ Note that $\prec$ and $\simeq$ and the term extension are generalized accordingly. Form this point on we will use the ordering relation modulo an equational theory when discussing generalization.

A term $t$ is called a \emph{generalization} or an
\emph{anti-instance} modulo an equational theory $\mathcal{E}$ of two terms $t_1$ and $t_2$ if $t\preceq_{\mathcal{E}} t_1$
and $t\preceq_{\mathcal{E}}  t_2$. It is a \emph{higher-order pattern generalization}
if additionally $t$ is a higher-order pattern. It is the \emph{least
  general generalization} (lgg in short), aka a \emph{most specific
  anti-instance,} of $t_1$ and $t_2$, if there is no generalization
$s$ of $t_1$ and $t_2$ which satisfies $t\prec_{\mathcal{E}}  s$. 
 
An \emph{anti-unification problem} (shortly AUP) is a triple
$X(\vect{x}): t \triangleq s$ where
\begin{itemize}
  \item $\lambda \vect{x}.X(\vect{x})$, $\lambda \vect{x}. t$, and
    $\lambda \vect{x}. s$ are terms of the same type,
  \item $t$ and $s$ are in $\eta$-long $\beta$-normal form, and
  \item $X$ does not occur in $t$ and $s$. 
\end{itemize}
 The variable $X$ is
called a \emph{generalization variable.}  The term
$X(\vect{x})$ is called the \emph{generalization term.}  The variables
that belong to $\vect{x}$, as well as bound variables, are written in
the lower case letters $x,y,z,\ldots$. Originally free variables,
including the generalization variables, are written with the capital
letters $X,Y,Z,\ldots$. This notation intuitively corresponds to the
usual convention about syntactically distinguishing bound and free
variables. The size of a set of AUPs is defined as
$\size{\{X_1(\vect{x_1}):t_1 \triangleq s_1,\dots,X_n(\vect{x_n}):t_n
  \triangleq s_n\}} = \sum_{i=1}^n \size{t_i} + \size{s_i}$. Notice 
that the size of $X_i(\vect{x_i})$ is not considered.
An \emph{anti-unifier} of an AUP $X(\vect{x}) :t\triangleq s$ is a
substitution $\sigma$ such that $\dom(\sigma) = \{X\}$ and $\lambda
\vect{x}.X(\vect{x})\sigma$ is a term which
generalizes both $\lambda \vect{x}.t$ and $\lambda \vect{x}.s$.

An anti-unifier of $X(\vect{x}):t\triangleq s$ is \emph{least general}
(or \emph{most specific}) modulo an equational theory $\mathcal{E}$ if there is no anti-unifier $\vartheta$ of
the same problem that satisfies $\sigma\prec_{\mathcal{E}}  \vartheta$. Obviously, if
$\sigma$ is a least general anti-unifier of an AUP
$X(\vect{x}):t\triangleq s$, then $\lambda \vect{x}.X(\vect{x})\sigma$
is a lgg of $\lambda \vect{x}.t$ and $\lambda
\vect{x}.s$.

Here we consider a variant of higher-order equational anti-unification problem: 
\begin{description}
  \item[Given:] Higher-order terms $t$ and $s$ of the same type in $\eta$-long
    $\beta$-normal form and an equational theory $\mathcal{E}\subseteq \{{\sf A},{\sf C}\}$.
  \item[Find:] A higher-order pattern generalization $r$ of $t$ and $s$ modulo $\mathcal{E}\subseteq \{{\sf A},{\sf C}\}$.
\end{description}

Essentially, we are looking for $r$ which is least
general among all higher-order patterns which generalize $t$ and
$s$ (modulo $\mathcal{E}$). There can still exist a term which is less general than $r$,
generalizes both $s$ and $t$, but is not a higher-order pattern. In \cite{DBLP:journals/jar/BaumgartnerKLV17} there is an instance for syntactic anti-unification: if $t=\lambda x,y. f(h(x,x,y),h(x,y,y))$ and $s=\lambda
x,y. f(g(x,x,y),$ $g(x,y,y))$, then $r=\lambda
x,y. f(Y_1(x,y),Y_2(x,y))$ is a higher-order pattern, which is an lgg
of $t$ and $s$. However, the term $\lambda x,y. f(Z(x,x,y),Z(x,y,y))$,
which is not a higher-order pattern, is less general than $r$ and
generalizes $t$ and $s$.

Another important distinguishing feature of higher-order pattern generalization modulo $\mathcal{E}$ is that there may be more than one least general pattern generalization (lgpg) for a given pair of terms. In the syntactic case there is a unique lgpg. The main contribution of this paper is to find conditions on the AUPs under which there is a unique lgpg for equational cases, and introduce weaker-optimality conditions which allow one to greedily search the space for a less general generalization compared to the syntactic one. We formalize these concepts in the following sections. 

\section{Higher Order Pattern Generalization in the Empty Theory}
\label{sect:hopau}

Below we assume that in the AUPs of the form $X(\vect{x}) : t \triangleq
s$ and the term $\lambda \vect{x}.X(\vect{x})$ is a higher-order pattern.  We now introduce the rules for the higher-order pattern generalization algorithm from~\cite{DBLP:journals/jar/BaumgartnerKLV17}, which works for $\mathcal{E}=\emptyset$. It produces syntactic higher-order pattern generalizations in linear time and will play a key role in our optimality conditions introduced in later sections. 

These rules work on triples $A;S;\sigma$, which are called \emph{states}. Here $A$ is a set of AUPs of the form $\{X_1(\vect{x_1}) : t_1 \triangleq s_1, \ldots, X_n(\vect{x_n}) : t_n \triangleq s_n\}$ that are pending to anti-unify, $S$ is a set of already solved AUPs (the \emph{store}), and $\sigma$ is a substitution (computed so far) mapping variables to patterns. The symbol $\uplus$ denotes disjoint union.
\infrule{Dec}{Decomposition}
    {\{X(\vect{x}) : h(t_1,\ldots,t_m) \triangleq h(s_1,\ldots,s_m)\}\uplus A ;\;  S ;\;  \sigma \Lra \\  
     \hspace*{3em} \{Y_1(\vect{x}): t_1\triangleq s_1,\ldots, Y_m(\vect{x}): t_m\triangleq s_m\}\cup A ;\;  S ;\;  
     \sigma\{X\mapsto \lambda \vect{x}.h(Y_1(\vect{x}),\ldots,Y_m(\vect{x}))\}, } 
[\noindent where $h$ is a free constant or $h\in \vect{x}$, and $Y_1,\ldots,Y_n$ are fresh variables of the appropriate types.]

\infrule{Abs}{Abstraction Rule}
    {\{\lbrace X(\vect{x}):\lambda y.t \triangleq \lambda z.s\rbrace\}\uplus A ;\;  S ;\;  \sigma \Lra \\ \hspace*{3em} \lbrace X'(\vect{x},y):t \triangleq s\lbrace z \mapsto y \rbrace \rbrace  \cup A ;\;   S ;\;  
     \sigma\left\lbrace X\mapsto  \lambda\vect{x},y.X'(\vect{x},y) \right\rbrace, }
[\noindent where $X'$  is a fresh variable of the appropriate type.]

\infrule{Sol}{Solve Rule}
    {\{X(\vect{x}):t \triangleq s\}\uplus A ;\;  S ;\;  \sigma \Lra  A ;\; \{Y(\vect{y}):t \triangleq s \}\cup S ;\;  
     \sigma\left\lbrace X\mapsto  \lambda\vect{x}.Y(\vect{y}) \right\rbrace}
[\noindent ce*{0.5em}
[\noindent where t and s are of a basic type, $\head(t) \neq \head(s)$ or $\head(t)  = \head(s) = Z \not \in \vect{x}$. The sequence $\vect{y}$ is a subsequence of $\vect{x}$ consisting of the variables that appear freely in $t$ or in $s$, and $Y$ is a fresh variable of the appropriate type.] 

\infrule{Mer}{Merge Rule}
    { A ;\; \lbrace X(\vect{x}):t_1\triangleq t_2 , Y(\vect{y}):s_1\triangleq s_2 \rbrace\uplus  S ;\;  \sigma \Lra \\ \hspace*{3em}   A ;\; \lbrace X(\vect{x}):t_1\triangleq t_2 \rbrace \cup S ;\;  
     \sigma\left\lbrace Y\mapsto  \lambda\vect{y}.X(\vect{x}\pi) \right\rbrace }
[\noindent Where $\pi:\lbrace\vect{x}\rbrace \rightarrow \lbrace\vect{y}\rbrace$ is a bijection, extended as a substitution with $t_{1}\pi = s_1$ and $t_2\pi = s_{2}$. Note that in the case of the equational theory we will consider later we would use $=_{\mathcal{E}}$ instead of $=$.]\vspace{0.3cm}

We will refer to these generalization rules as $\mathcal{G}_{base}$. To compute generalizations for two simply typed lambda-terms in $\eta$-long $\beta$-normal form $t$ and $s$, the algorithm from \cite{DBLP:journals/jar/BaumgartnerKLV17} starts with the
\emph{initial state} $\{X : t \triangleq s\}; \emptyset; \ids$, where
$X$ is a fresh variable, and applies these rules as long as possible. The computed result  is the instance of $X$ under the final substitution. It is the syntactic least general higher-order pattern generalization of $t$ and $s$, and is computed in linear time in the size of the input.

We will use this linear time procedure in the following section to obtain ``optimal'' least general higher-order pattern generalizations of terms modulo an equation theory. These optimal generalizations are dependent on the generalizations the syntactic algorithm produces. When we need to check more than one decomposition of a given AUP in order to compute the optimal generalizations modulo an equational theory, we  compute the optimal generalization for each decomposition path and than compare the results. The details are explained below.

We assume that terms are written in \emph{flattened form}, obtained by replacing all subterms of the form$f(t_1,\ldots,f(s_1,\ldots,s_m),\ldots t_n)$ by $f(t_1,\ldots,s_1,\ldots,s_m,\ldots t_n)$, where ${\sf A}\in Ax(f)$. Also, by convention, the term $f(t)$ stands for $t$, if ${\sf A}\in Ax(f)$.

\section{Equational Decomposition Rules}
\label{sect:genAC}

In this section we discuss an extension of the basic rules concerning higher-order pattern generalization by decomposition rules for A, C, and AC function symbols. Here, we consider the general, unrestricted case. Efficient special fragments are discussed in the subsequent section.

We start from decomposition rules for associative generalization: 

\infrule{Dec-A-L}{Associative Decomposition Left}
    {\{X(\vect{x}) : f(t_1,\ldots,t_n) \triangleq f(s_1,\ldots,s_m)\}\uplus A ;\;  S ;\;  \sigma \Lra \\  
     \hspace*{3em} \{Y_1(\vect{x}): f(t_1,\ldots, t_k )\triangleq s_1 ,\; Y_2(\vect{x}): f(t_{k+1},\ldots, t_m)\triangleq f(s_2,\ldots,s_m)\}\cup A ;\;  \\ \hspace*{3em} S ;\;  
     \sigma\{X\mapsto \lambda \vect{x}.f(Y_1(\vect{x}),Y_2(\vect{x}))\},}
     [\noindent where $\mathit{Ax}(f) = \lbrace {\sf A}\rbrace$, $1\le k\le n-1$, $n,m\geq 2$, and $Y_1$ and $Y_2$ are fresh variables of appropriate types.]
\infrule{Dec-A-R}{Associative Decomposition Right}
    {\{X(\vect{x}) : f(t_1,\ldots,t_n) \triangleq f(s_1,\ldots,s_m)\}\uplus A ;\;  S ;\;  \sigma \Lra \\  
     \hspace*{3em} \{Y_1(\vect{x}): t_1\triangleq f(s_1,\ldots, s_k ) , Y_2(\vect{x}): f(t_2,\ldots, t_n )\triangleq f(s_{k+1},\ldots, s_m)\}\cup A ;\;  \\ \hspace*{3em} S ;\;  
     \sigma\{X\mapsto \lambda \vect{x}.f(Y_1(\vect{x}),Y_2(\vect{x}))\},}
[\noindent where  $\mathit{Ax}(f) = \lbrace {\sf A}\rbrace$, $1\le k\le m-1$, $n,m\ge 2$, and $Y_1$ and $Y_2$ are fresh variables of appropriate types. ]\vspace*{1em}

We refer to the extension of $\mathcal{G}_{base}$ by the above associativity rules as $\mathcal{G}_A$ and extend the termination, soundness and completeness results for  $\mathcal{G}_{base}$ to $\mathcal{G}_A$. 
\begin{theorem}[Termination]
The set of transformations $\mathcal{G}_A$ is terminating.
\end{theorem}
\begin{proof}
Termination follows from the fact that $\mathcal{G}_{base}$ terminates~\cite{DBLP:journals/jar/BaumgartnerKLV17} and the rules {\sf Dec-A-L} and {\sf Dec-A-R} can be applied finitely many times. 
\end{proof}
\begin{theorem}[Soundness]
 If $\{X: t\triangleq s\};\emptyset;\ids \Lra^* \emptyset; S;\sigma$
   is a transformation sequence of $\mathcal{G}_A$, then $X\sigma$ is a higher-order pattern in $\eta$-long $\beta$-normal form and $X\sigma \preceq t$ and $X\sigma \preceq s$.
\end{theorem}
\begin{proof}
It was shown in~\cite{DBLP:journals/jar/BaumgartnerKLV17} that $\mathcal{G}_{base}$ is sound. Let us assume as a base case that all occurrences of associative function symbols in $t\triangleq s$ have two arguments. Then the rules {\sf Dec-A-L} and {\sf Dec-A-R} are equivalent to the $Dec$ rule. As an induction hypothesis (IH), assume soundness holds when  all occurrences of associative function symbols in $t\triangleq s$ have  $\leq n$ arguments. We show that it holds for $n+1$. Let $t\triangleq s$ be of the form $f(t_{1},\ldots, t_m) \triangleq f(s_{1},\ldots, s_k)$ for $\max\{k,m\}\leq (n+1)$ and let associative function symbols occurring in $t_{1},\ldots t_m,s_{1},\ldots s_k$ have at most $n$ arguments. Any application of {\sf Dec-A-L} or {\sf Dec-A-R} will produce two AUPs  for which the IH holds, and thus, the theorem holds. We can extend this argument to an arbitrary number of associative function symbols with $n+1$ arguments with another induction.  
\end{proof}

\begin{theorem}[Completeness]
   Let $\lambda\vect{x}.t_1$ and $\lambda\vect{x}.t_2$ be higher-order
   terms and $\lambda\vect{x}.s$ be a higher-order pattern such that
   $\lambda\vect{x}.s$ is a generalization of both
   $\lambda\vect{x}.t_1$ and $\lambda\vect{x}.t_2$ modulo associativity. Then there exists
   a transformation sequence $\{X(\vect{x}): t_1\triangleq
   t_2\};\emptyset;\ids \Lra^* \emptyset; S;\sigma$ in $\mathcal{G}_{A}$ such
   that $\lambda\vect{x}.s \preceq X\sigma$.
\end{theorem}
\begin{proof}
We can reason similarly to the previous proof. It was shown in~\cite{DBLP:journals/jar/BaumgartnerKLV17} that $\mathcal{G}_{base}$ is complete. Let us assume as a base case that all occurrences of associative function symbols in $t\triangleq s$ have two arguments. Then the rules {\sf Dec-A-L} and {\sf Dec-A-R} are equivalent to the $Dec$ rule and completeness holds. When we have $n+1$ arguments there are $ n$ ways to group the arguments associatively and the decompositions rules {\sf Dec-A-L} and {\sf Dec-A-R} allow one to consider all groupings.
\end{proof}
 
The addition of associative function symbols allows for  more than one decomposition and thus more than one lgg in contrast to higher-order pattern generalization which results in a unique lgg . If we wish to compute the complete set of lggs we would simply exhaust all possible applications of the above rules. However, for most applications an ``optimal'' generalization is sufficient. We postpone discussion  till the next section.

The decomposition rule for commutative symbols is also pretty intuitive:
\infrule{Dec-C}{Commutative Decomposition}
     {\{X(\vect{x}) : f(t_1,t_2) \triangleq f(s_1,s_2)\}\uplus A ;\;  S ;\;  \sigma \Lra \\  
     \hspace*{3em} \{Y_1(\vect{x}): t_1\triangleq s_i ,\ Y_2(\vect{x}): t_2\triangleq s_{(i\mod 2)+1}\}\cup A ;\;  S ;\;  
     \sigma\{X\mapsto \lambda \vect{x}.f(Y_1(\vect{x}),Y_2(\vect{x}))\},}
     [\noindent where $\mathit{Ax}(f) = \lbrace {\sf C}\rbrace$, $i\in \lbrace 1, 2\rbrace$, and $Y_1$ and $Y_2$ are fresh variables of appropriate types.]\vspace{0.3cm}

We refer to the extension of $\mathcal{G}_{base}$ by the commutativity rule as $\mathcal{G}_C$. We can easily extend the termination, soundness, and completeness results to  $\mathcal{G}_\sC$. Notice that also for commutative generalization, the lgg is not necessarily unique. 

Unlike commutativity, which considers a fixed number of terms, and associativity, which enforces an ordering on terms, AC function symbols allow an arbitrary number of arguments with no fixed ordering on the terms. The corresponding decomposition rules take it into account:

\infrule{{\sf Dec-AC-L}}{Associative-Commutative Decomposition Left}
    {\{X(\vect{x}) : f(t_1,\ldots,t_n) \triangleq f(s_1,\ldots,s_m)\}\uplus A ;\;  S ;\;  \sigma \Lra \\  
     \hspace*{3em} \{Y_1(\vect{x}):h(t_{i_{1}},\ldots, t_{i_{l}} )\triangleq s_{k} , \ 
      Y_2(\vect{x}):f(t_{i_{(l+1)}},\ldots, t_{i_n})\triangleq f(s_1,\ldots,s_{k-1},$ $s_{k+1},\ldots, s_{m})\}\cup A ;\;  \\ \hspace*{3em} S ;\;  
     \sigma\{X\mapsto \lambda \vect{x}.f(Y_1(\vect{x}),Y_2(\vect{x}))\},}
[\noindent where $\mathit{Ax}(f) = \lbrace \sf A,C\rbrace$, $\lbrace i_{1},\ldots, i_n \rbrace \equiv \lbrace 1,\ldots, n\rbrace$, $l\in \lbrace 1,\ldots, n-1\rbrace$,  $k\in \lbrace 1,\ldots, m\rbrace$,  $n,m\geq 2$, and $Y_1$ and $Y_2$ are fresh variables of appropriate types.]

\infrule{{\sf Dec-AC-R}}{Associative-Commutative Decomposition Right}
    {\{X(\vect{x}) : f(t_1,\ldots,t_n) \triangleq f(s_1,\ldots,s_m)\}\uplus A ;\;  S ;\;  \sigma \Lra \\  
     \hspace*{3em} \{Y_1(\vect{x}):t_k\triangleq f(s_{i_{1}},\ldots, s_{i_{l}}) ,\ 
      Y_2(\vect{x}):f(t_1,\ldots,t_{k-1},t_{k+1},\ldots, t_{m}) \triangleq f(s_{i_{(l+1)}},\ldots, s_{i_n}) \}\cup A ;\;  \\ \hspace*{3em} S ;\;  
     \sigma\{X\mapsto \lambda \vect{x}.f(Y_1(\vect{x}),Y_2(\vect{x}))\},}
[\noindent where $\mathit{Ax}(f) = \lbrace\sf  A,C\rbrace$, $\lbrace i_{1},\ldots, i_n \rbrace \equiv \lbrace 1,\ldots, m\rbrace$, $l\in \lbrace 1,\ldots, m-1\rbrace$,  $k\in \lbrace 1,\ldots, n\rbrace$,  $n,m\geq 2$, and $Y_1$ and $Y_2$ are fresh variables of appropriate types.]\vspace{0.3cm}

We refer to the extension of $\mathcal{G}_{base}$ by the AC decomposition rules as $\mathcal{G}_{\sA\sC}$. Again, termination, soundness and completeness are easily extended to this case.

\section{Towards Special Fragments}
\label{sec:special:fragments}

This section is devoted to computing special kind of ``optimal'' generalizations, which can be done more efficiently than the general unrestricted cases considered in the previous section. 

The idea is the following: The equational decomposition rules introduce branching in the search space. Each branch can be developed in linear time, but there can be too many of them. However, if the branching factor is bounded, we could choose one of the alternative states (produced by decomposition) based on some ``optimality'' criterion, and develop only that branch. Such a greedy approach will give one ``optimal'' generalization.

In order to have a ``reasonable'' complexity, we should be able to choose such an optimal state from ``reasonably'' many alternatives in ``reasonable'' time. For this, our idea is to treat all the alternative states obtained by an equational decomposition step as syntactic anti-unification problems, compute lggs for each of them (which can be done in linear time), choose the best one among those lggs (e.g., less general than the others, or, if there are several such results, use some heuristics), and restart equational anti-unification algorithm from the state which led to the computation of that best syntactic lgg. When the branching factor is constant, this leads to a quadratic algorithm, and when it is linearly bounded, we get a cubic algorithm. These are the cases we consider below. We would also need to decompose in a more clever way than in the rules above, where the decomposition was based on an arbitrary choice of a subterm.

Hence, we need to identify fragments of equational anti-unification problems which would have the decomposition branching factor constant or linearly bounded.  We start by introducing the following concepts.

\begin{definition}[$E$-refined generalization]
Given two terms $t$ and $s$ and their $\cE$-generalizations $r$ and $r'$, we say that $r$ is \emph{at least as good as} $r'$ with respect to $\cE$ if either $r'\preceq_\cE r$ or they are not comparable with respect to $\preceq_\cE$.

An $\cE$-generalization $r$ of $t$ and $s$ is called their {\em $E$-refined generalization} iff $r$ is at least as good (with respect to $\cE$) as a syntactic lgg of $t$ and $s$.
\end{definition}

Note that every syntactic generalization is also an $\cE$-generalization. A direct consequence of this definition is that every element of the minimal complete set of $\cE$-generalizations (where $\cE$ is A, C, or AC) of two terms  is an $\cE$-refined generalization of $t$ and $s$. However, there might exist $\cE$-refined generalizations which do not belong to the minimal complete set of generalizations.

Looking back at the informal description of the construction above, we can say that at each branching point we will be aiming at choosing the alternative that would lead to ``the best'' $\cE$-refined generalization.

The concept of $E$-refined allows us to compute better generalizations than the base procedure would do, without concerning ourselves with certain difficult to handle decompositions. We will outline what we mean by ``difficult'' in later sections. Some of these difficult decompositions can be handled by finding {\em alignments} between two sequences of terms. 

\begin{definition}[Alignment, Rigidity Function]
Let $w_1$ and $w_2$ be strings of symbols.
Then the sequence $a_1 [i_1 , j_1 ] \cdots a_n [i_n , j_n ]$, for $n \geq 0$ and $a_k$ are not variables, is an {\em alignment} if
\begin{itemize}
\item $i$’s and $j$’s are integers such that $0 < i_1 < \cdots < i_n < |w_1 |$ and $0 < j_1 <
\cdots < j_n < |w_2 |$, and
\item $a_k = w_1 |_{i_k} = w_2|_{j_k}$ , for all $1 \leq k \leq n$. An alignment of the form $a_1 [i, j]$ will be referred to as a {\em singleton alignment}
\end{itemize}

The set of all alignments will be denoted by $\textbf{A}$. A {\em (singleton) rigidity function} $\mathcal{R}$ is a function that returns, for every pair of strings of symbols
$w_1$ and $w_2$, a set of (singleton) alignments of $w_1$ and $w_2$.
\end{definition}

\begin{definition} [Pair of argument head sequences and multisets]
Let $t= f(t_{1},\ldots,t_{n})$ and $s= f(s_{1},\ldots,s_{m})$. Then the \emph{pair of argument head sequences} and the \emph{pair of argument head multisets} of $t$ and $s$, denoted respectively as $\pahs(t,s)$ and $\pahm(t,s)$, are defined as follows:
\begin{alignat*}{1}
 & \pahs(t,s)=  \left\langle (\mathit{head}(t_{1}),\ldots,\mathit{head}(t_{n})),\ (\mathit{head}(s_{1}),\ldots,\mathit{head}(s_{m})) \right\rangle.\\
 & \pahm(t,s)=  \left\langle \ldblbrace \mathit{head}(t_{1}),\ldots,\mathit{head}(t_{n})\rdblbrace, \ \ldblbrace\mathit{head}(s_{1}),\ldots,\mathit{head}(s_{m})\rdblbrace \right\rangle.
\end{alignat*}

These notions extend to AUPs: A pair of argument head sequences (resp. multisets) of an AUP $X(\vect{x}):t \triangleq s$ is the pair of argument head sequences (resp. multisets) of the terms $t$ and $s$.
\end{definition}

There is a subset of AUPs, referred to as {\em $1$-Determined AUPs}, which contain associative function symbols and have an interesting $\mathcal{E}$-refined generalizations are computable in linear time.  The more general $r$-determined AUPs allow a bounded number of possible choices, that is $r$ choices, whenever associative decomposition may be applied. Even for $2$-determined AUPs computing the set of lggs is of exponential complexity. Therefore, we introduce the notion of {\em $(\mathcal{R},C,\mathcal{G})$-optimal generalization} where $\mathcal{R}$ is a so called rigidity function~\cite{DBLP:journals/jar/KutsiaLV14} and $C$ is a choice function picking one of available decompositions. Under such optimality conditions, we are able to compute an $\mathcal{E}$-refined generalization in quadratic time for $k$-determined AUPs and in cubic time for arbitrary AUPs with associative function symbols.

The equational decomposition rules above are too non-deterministic and the computed set of generalizations has to be minimized to obtain minimal complete sets of generalizations. However, even if we performed more guided decompositions, obtaining e.g., terms with the same head in new AUPs (as in \cite{DBLP:journals/jar/KutsiaLV14}), there would still be alternatives. For instance, consider the following AUP where $f$ is associative: 
$ X(\vect{x}):f(t_{1}, \ldots t_{i},\ldots, t_{j},\ldots, t_n)\triangleq f(s_{1}, \ldots s_{i},\ldots, s_{j},\ldots,  s_m).$
Now let $head(t_{i}) = head(s_{j})$, $head(s_{i}) = head(t_{j})$, and for every other term comparison whose index is $\leq j$ the head symbols are not equivalent. Under these assumptions there is not enough information to decide which decomposition is less general. Furthermore, this can be generalized from two possible decompositions to $k$ possibilities. 

Under certain conditions we can force a term to have a single decomposition path, what we will refer to as a {\em $1$-determined} condition which is equivalent to unique longest common subsequence of head symbols. We formally define $k$-determined AUPs using the following sequence of definitions:

\begin{definition}[$k$-determinate set]
\label{def:determined:set}
Given the pair of sequences of symbols $\langle s_1, s_2 \rangle$ with $s_1=(a_1,\ldots,a_n)$ and  $s_2=(b_1,\ldots,b_m)$, and a positive integer $k$, the \emph{(strict) $k$-determinate set} of $s_1$ and $s_2$, denoted $\kdet{k}{s_1}{s_2}$ ($\skdet{k}{s_1}{s_2}$), is defined as follows:
\begin{itemize}
\item If $n=0$ and $m\not =0$ or vice versa, then $\kdet{k}{s_1}{s_2}=\emptyset $.
\item Otherwise, let $1\le i \le \min(n,m)$ be a number such that for the multiset $M_i= (\ldblbrace a_1 \rdblbrace \cap \ldblbrace b_1 \rdblbrace) \cup (\ldblbrace a_2,\ldots, a_i\rdblbrace \cap \ldblbrace b_2,\ldots, b_i\rdblbrace )\neq \emptyset$ we have $M_{i} \cap \ldblbrace b_{i+1},\ldots, b_m\rdblbrace = M_{i} \cap \ldblbrace a_{i+1},\ldots, a_n\rdblbrace = \emptyset$. Let $K$ ($K_{s}$) be the set of pairs $\{ a_{j_1}[j_1,j_2] \mid a_{j_1} = b_{j_2} \text{ and } j_1=1 \text{ iff } j_2=1 \}$ ( $\{ a_{j_1}[j_1,j_2] \mid a_{j_1} = b_{j_2} \text{ and }  \}$). If $K$ has at most $k$ elements, then
\[det(k, s_1,s_2):= \bigcup_{a_{j_1}[j_1,j_2]\in K}  \mathit{add}(a_{j_1}[j_1,j_2],\kdet{k}{ (a_{j_1+1},\ldots,a_n)}{ (b_{j_2+1},\ldots,b_m)}) . \]

 \[\mathit{add}(a,A) = \left\lbrace \begin{array}{c|c} \left\lbrace (a,A)\right\rbrace  &  A \not = \emptyset\\
 \emptyset & otherwise  \end{array} \right.\]
\item Otherwise, $det(k, s_1,s_2)=\left\lbrace \emptyset\right\rbrace$ .
\end{itemize}
Note that $\skdet{k}{s_1}{s_2}$ is defined analogously using $K_{s}$ instead of $K$.
We will refer to the pairs $(a,A)$ where $a$ is a singleton alignment and $A$ a $k$-determinate set as {\em blocks}.
\end{definition}

We will use $\skdet{k}{s_1}{s_2}$ when considering commutativity in Section~\ref{sect:special:C}
.
\begin{example} We illustrate the previous definition:
\begin{itemize}
\item $\kdet{1}{(a,b)}{(a,b)} = \{\block{a[1,1]}{\block{b[2,2]}{\emptyset}} \}$.
\item $\kdet{1}{(a,a)}{(b,a)} = \{ \block{\{a[2,2]}{\emptyset} \}$.
\item $\kdet{1}{(a,c,c,b,a,c)}{(a,d,b,a,c)} = \{ \block{a[1,1]}{\block{b[4,3]}{\block{a[5,4]}{\block{c[6,5]}{\emptyset}}}}\}$.
\item $\kdet{1}{(a,b,a)}{(c,a,b,c)} = \left\lbrace \emptyset \right\rbrace $
\item $\kdet{1}{(a,b,d)}{(c,a,b,c)} = \{\block{b[2,3]}{\emptyset}\}$
\item $\kdet{2}{(a,b,a)}{(c,a,b,c)} = \left\lbrace \block{b[2,3]}{\emptyset}\right\rbrace $
\item $\skdet{1}{(a,b)}{(b,a)} = \left\lbrace \block{a[1,2]}{\emptyset}, \block{b[2,1]}{\emptyset} \right\rbrace $
\item $\kdet{2}{(c,a,b,c)}{(d,b,a,d)} = \{ \block{a[2,3]}{\emptyset} ,\block{b[3,2]}{\emptyset} \}$.
\item $\kdet{3}{(a,b,a,c,d)}{(c,a,b,a,d)} = \\ \ \ \ \{ \block{b[2,3]}{ \block{a[1,1]}{\emptyset}} ,\block{a[3,2]}{\block{d[2,3]}{\emptyset}},\block{a[3,4]}{\emptyset} \}$.
\item $\kdet{k}{(a,a)}{(b,c,d)} = \left\lbrace \emptyset \right\rbrace $.
\item $\kdet{k}{(a,b)}{(a)} =  \emptyset  $.
\item $\kdet{k}{(a,a)}{(a)} = \left\lbrace \emptyset \right\rbrace  $.

\end{itemize}
Even though $\kdet{k}{(a,b)}{(a)}$ are related $\kdet{k}{(a,a)}{(a)}$ the formalism does not handle them as similar. This merely makes the formalism a little more restricted.  Notice that a unique longest common subsequence of two symbol sequences is not equivalent to k-determined. Consider the following example:
\begin{itemize}
\item $\kdet{k}{(c,a,a,d)}{(c,a,b,a,d)} = \{ \block{c[1,1]}{ \block{a[1,1]}{\block{d[2,3]}{\emptyset}}} \}.$
\end{itemize} 
The alignment representing its longest common subsequence is
\begin{itemize}
\item $c[1,1]a[2,2]a[3,4]d[4,5]$
\end{itemize} \end{example}

\begin{definition}[$k$-determined term pairs]
A pair of terms $\langle t, s\rangle$ is $k$-determined iff either $\head(t)\neq \head(s)$ or $\head(t)= \head(s)=f$ and
\begin{itemize}
\item $Ax(f)=\emptyset$, or
\item $Ax(f)=\{\sA \}$ and $\kdet{k}{\pahs(t}{s)}\not = \emptyset$ , or
\end{itemize}
Furthermore we say that the pair $\langle t, s\rangle$ is {\em total $k$-determined} if $t=\lambda x_1,\ldots, x_n.t' $, $s=\lambda y_1,\ldots, y_n.s'$ and $|t'| = |s'| =1$ or for each $(a[i,j],S) \in \kdet{k}{\pahs(t}{s)}$ where $t_i$ is the term at the $i^{\mbox{th}}$ position of $t$ and $s_j$ is the term at the $j^{\mbox{th}}$ position of $s$ the term pair  $\langle t_i, s_j\rangle$ is total $k$-determined.
\end{definition}

 \begin{proposition}
 \label{prop:complexantilong}
The complexity of checking if the terms of an  AUP $X(\bar{x}): \lambda x_{1},\ldots, x_{l}.f(t_{1},\ldots,t_{n}) \triangleq \lambda y_{1},\ldots, y_{k}.f(s_{1},\ldots,s_{m})$ is $1$-determined is $O(n)$ and total $1$-determined is $O(n^2)$, where $n$ is maximum of the length of the two terms. 
 \end{proposition}
 
Checking $k$-determinedness of an AUP  is a harder problem complexity-wise. For example, given the sequences $(a,\ldots,a)$ and $(a,\ldots,a)$ there are $n^2$ ways to align the terms which have to be checked. Moreover, if we want to check total $k$-determinedness  we have to again do a quadratic check for each pair of aligned terms resulting in an $O(n^4)$ procedure. 

\section{Associative Generalization: Special Fragments and Optimality}
\label{sect:special:A}
\subsection{Associativity and $1$-Determined AUPs}
We provide a linear time algorithm for higher-order $\{{ \sf A}\}$-refined pattern generalization of AUPs which are $1$-determined. Essentially, at every step there is a single decomposition choice which can be made. 
\begin{theorem}
\label{thm:AUPAnti}
A higher-order $\{{ \sf A}\}$-refined pattern generalizer for a {\em total $1$-determined} AUP can be computed in linear time.
\end{theorem}
\begin{proof}
If the AUP does not contain an associative function symbol, then its $E$-refined generalization, which is also an lgg, can be computed in linear time~\cite{DBLP:journals/jar/BaumgartnerKLV17}. If it does contain an associative function symbol, we have two alternatives: either every occurrence of the associative function symbol has two arguments (remember that our terms are in flattened form), or not. In the former case, the associative decomposition rules do not differ from the syntactic decomposition rule {\sf Dec} and we can only apply the latter. It means that we can still use the linear algorithm from~\cite{DBLP:journals/jar/BaumgartnerKLV17}. The rest of the proof is about the case when there are occurrences of associative function symbols with more than two arguments. The proof goes by induction on the maximal number of such arguments.

We assume for the induction hypothesis that if every instance of the associative function symbol in the AUP has at most $n$ arguments, then it is solvable in linear time, and show that the same holds for $n+1$. Let us assume that the AUP we are currently considering has the following form $X(\vect{x}):f(t_{1}, \ldots, t_m)\triangleq f(s_{1}, \ldots, s_{k})$ where $f$ is associative and $\max\lbrace m,k\rbrace= n+1$. Assume without loss of generality that $k=n+1$. Also, assume that no other occurrence of $f$ in the given AUP has more than $n$ arguments. We make this assumption in order to reduce the complexity of associative decomposition in the AUP and thus, apply the induction hypothesis. If   $\head(t_{1}) = \head(s_{1})$,then their lgg should not be a variable. Therefore, we can apply {\sf Dec-A-L}, which results in the AUPs $X(\vect{x}): t_{1} \triangleq s_{1}$ (whose further decomposition will make sure that they $t_1$ and $s_1$ are not generalized by a generalization variable) and $X(\vect{x}):f(t_{2}, \ldots, t_m)\triangleq f(s_{2}, \ldots, s_{n+1})$. Notice that both of the resulting  AUPs, by our assumptions, only contain $f$ with not more than $n$ arguments. Thus, by the induction hypothesis the theorem holds in this case. 

For the next step we assume $s$ and $t$ are the terms of the AUP and that $(h[l,l],S)\in \kdet{1}{\pahs(t}{s)}$ s.t. $Ax(h)= \left\lbrace A\right\rbrace$.  Therefore, we can perform {\sf Dec-A-L} only on the first argument $l-1$ times, which gives the following new AUPs: $\{X_1(\vect{x}): t_{1} \triangleq s_{1}, \ldots, X_{l-1}(\vect{x}): t_{l-1} \triangleq s_{l-1},\, X_{l}(\vect{x}):f(t_{l} \ldots, t_m)\triangleq f(s_{l}, \ldots, s_{n+1}) \}$. 
All the resulting  AUPs, by our assumptions, only contain $f$ with not more than 
$n$ arguments, thus by the induction hypothesis the theorem holds in this case. 

For the next step we assume $s$ and $t$ are the terms of the AUP and that $(h[i,j],S)\in \kdet{1}{\pahs(t}{s)}$ s.t. $Ax(h)= \left\lbrace A\right\rbrace$ and $i\not = j$. This is similar to the previous case except there is more than one possible way to apply associative decomposition. More precisely, the number of possible ways is $F(l-j+1)$ where 
\begin{alignat*}{1}
& F(0) = 1, \qquad \qquad F(r+1) = \sum_{w=1}^{r+1} F(r+1-w) \quad \text{for }r\ge 0.
\end{alignat*}
%
which is roughly $F(r) = 2^{(r-1)}$. However, being that none of the head symbols of obtained term-pairs are equivalent nor can their head symbols be equivalent to $f$, we know that none of the resulting AUPs will require further decomposition. Thus, we need to apply associative decomposition. This can be easily performed be performed by some heuristic. The result will be a set of AUPs containing $X(\vect{x}):f(t_{j} \ldots t_m)\triangleq f(s_{l}, \ldots s_{n+1})$ and thus by the induction hypothesis and our assumptions, the theorem holds in this case. 

For the final step we just need to apply a simple induction argument on the number of times in a term the associative symbol $f$ occurs with arity $n+1$. The above argument provides the step case and base case being that we prove the theorem for one occurrence and can use the proof for $p$ occurrences. Thus, the theorem holds.  
\end{proof}

 In the next section we consider AUPs which are  $k$-determined for $k>1$. This will require us introducing a new concept of optimality based on a choice function greedily applied during decomposition. 
 \subsection{Choice Functions and Optimality}
In this section procedures and optimality conditions for total $k$-determined AUPs, for $k>1$, that is AUPs where there are at most $k$ ways to apply equational decomposition. 

If we were to compute the set of $E$-refined generalizations for a total $k$-determined  AUP by testing every decomposition, even for $k=2$ the size of search space is too large to deal with efficiently. However, we can find a {\em $(\mathcal{R},C,\mathcal{G})$-optimal $E$-refined generalization} (precisely defined below) in quadratic time, where $\mathcal{R}$ is a singleton rigidity function, $C$ a $\mathcal{R}$-choice function, $\mathcal{G}$ is a set of state transformation rules. Essentially, $(\mathcal{R},C,\mathcal{G})$-optimality means the $\mathcal{R}$-choice function chooses the ``right'' computation path via $\mathcal{G}$ based on the singleton rigidity function $\mathcal{R}$. The effect is that we reduce the problem of total $k$-determined AUPs to the case of total $1$-determined AUPs with the additional complexity of computing the choice function at each step. We will provide a choice function with linear time complexity based on the procedure for $\mathcal{G}_{base}$.

We will denote the set of all AUPs by $\mathbb{A}$. We will need the concept for the following definitions. 

\begin{definition}[$(P,a)$-decomposition]
Let $P\equiv X(\bar{x}): \lambda x_{1},\ldots, x_{l}.f(t_{1},\ldots,t_{n}) \triangleq \lambda y_{1},\ldots, y_{k}.f(s_{1},$ $\ldots,s_{m})$, $a$ is an alignment of $\left\langle w_{1} , w_{2} \right\rangle_P$. An {\em $(P,a)$-decomposition of $P$} is $dec(P,a) =  \lbrace  Y_{(i,j)}(\vect{y}_{(i,j)}):t_{i}\triangleq s_{j }\mid h[i,j]\in a\ \rbrace $, where $Y_{(i,j)}$ are new variables of appropriate type and $\vect{y}_{(i,j)}$ are bound variables from $\vect{x}$, which appear in $t_{i}\triangleq s_{j}$.
\end{definition}
\begin{definition}[$\mathcal{G}$-feasible]
Let $A;S;\sigma$ be a state s.t. $P\in A$ where $P\equiv X(\vect{x}):\lambda x_{1},\ldots, x_{l}.f(t_{1},\ldots,$ $t_{n}) \triangleq \lambda y_{1},\ldots, y_{k}.f(s_{1},\ldots,s_{m})$, $a$ be an alignment of $\left\langle w_{1}, w_{2} \right\rangle_P$ and $\mathcal{G}_{base}\subseteq\mathcal{G}$ be a set of state transformation rules. We say that $dec(P,a)$ is {\em $\mathcal{G}$-feasible} if there exists $A;S;\sigma\Longrightarrow^{*} A';S';\sigma'$ using $\mathcal{G}$ such that $ A'= (A\setminus P) \cup dec(P,a)$. 
\end{definition}

\begin{definition}[$(\mathcal{R},P,\mathcal{G})$-branching]
Let $P\equiv X(\vect{x}):\lambda x_{1},\ldots, x_{l}.f(t_{1},\ldots,t_{n}) \triangleq \lambda y_{1},\ldots, y_{k}.$ $f(s_{1},\ldots,s_{m})$,  $\left\langle w_{1} , w_{2} \right\rangle_P$ be its pair of argument head sequences, $\mathcal{R}$ be a singleton rigidity function, and $\mathcal{G}_{base}\subseteq \mathcal{G}$ be a set of state transformation rules. An {\em $(\mathcal{R},P,\mathcal{G})$-branching} is a set $\mathit{B}(\mathcal{R},P) = \left\lbrace  \mathit{dec}(P,a) \ \vert \ a\in \mathcal{R}(w_1,w_2) \mbox{ and }\ \mathit{dec}(P,a)\ \right. $ $\left.  \mbox{is } \mathcal{G}\mbox{-feasible} \right\rbrace$. 
\end{definition} 
\begin{definition}[$\mathcal{R}$-Choice function]
Let $\mathcal{R}$ be a singleton rigidity function and $\mathcal{G}_{base}\subseteq \mathcal{G}$ be a set of state transformation rules. An $\mathcal{R}$-choice function $C_{(\mathcal{R},\mathcal{G})}:\mathbb{A} \rightarrow \textbf{A}$ is a partial function from AUPs to alignments such that if for some $P\in \mathbb{A}$ , $C_{(\mathcal{R},\mathcal{G})}(P) = a$, then $\mathit{dec}(P,a)\in \mathit{B}(\mathcal{R},P)$.
\end{definition}

\begin{definition}[$(\mathcal{R},C,\mathcal{G})$-optimal generalization]
\label{ouroptimal}
Let $A$ be $\{X(\bar{x}):t\triangleq s \}$, $\mathcal{R}$ be a singleton rigidity function, $C$ be an $\mathcal{R}$-choice function, and $\mathcal{G}_{base}\subseteq \mathcal{G}$ be a set of state transformation rules, which compute generalizations. We say that a generalization $k$ of the terms $t$ and $s$ is an $(\mathcal{R},C,\mathcal{G})$-optimal generalization if $r=X\sigma$, where $\sigma$ is  resulting from the derivation $A;\emptyset;\emptyset \Longrightarrow^{*} \emptyset ;S;\sigma$ using the rules of $\mathcal{G}$, in which  every decomposition is either syntactic or respects $C$.
\end{definition}

In the following subsection we show how the above definitions can lead to a more general result (compared to the one in the previous section) concerning associative generalization. 

\subsection{$k$-Determined Associative Generalization}
\label{subsec:RD}
Before defining our concrete choice function, we must define the singleton rigidity function we will use. Intuitively, it should select alignments from prefixes of involved sequences. The prefixes are of the same length and should be maximal among those that contain at most $k$ common elements. Formally, it is defined as follows:

\begin{definition}
Let $w_{1} =  (a_1,\ldots, a_n)$ and  $w_{2} = (b_1,\ldots, b_m)$ be sequences of symbols and $k\ge 1$ be an integer. We define the singleton rigidity function $\mathcal{R}_{\sf A}^{k}$ as  
\begin{equation}
\mathcal{R}_{\sf A}^{k}(w_1,w_2) = \left\lbrace \begin{array}{c|c} 
\left\lbrace a_{l}\left[ l,k\right]\ \middle\vert\ (a_{l}\left[ l,k\right],S)\in \kdet{k}{w_{1}}{w_{2}} \right\rbrace  & \begin{array}{c} \kdet{k}{w_{1}}{w_{2}}\not = \emptyset\\\end{array}\\
\emptyset & otherwise
\end{array} \right. 
\end{equation}
\end{definition}

Now we define a choice function taking an arbitrary singleton rigidity function. 

\begin{definition}
Let $P\equiv X(\vect{x}):\lambda x_{1},\ldots, x_{l}.f(t_{1},\ldots,t_{n}) \triangleq \lambda y_{1},\ldots, y_{k}.$ $f(s_{1},\ldots,s_{m})$ be an AUP and $f$ a function symbol such that $Ax(f)\not \equiv \emptyset$.  We define the choice function $C_{(\mathcal{R},\mathcal{G})}$, where $\mathcal{R}$ is a singleton rigidity function, and $\mathcal{G}$ is a set of state transformation rules containing $\mathcal{G}_{base}$, as follows:
\begin{equation}
C_{(\mathcal{R},\mathcal{G})}(P) = \left\lbrace \begin{array}{c|l} 
a_{\min} & \quad \mathit{B}(\mathcal{R},P)\not\equiv \emptyset \\
\textbf{undef} & \quad otherwise
\end{array} \right. 
\end{equation}
where  $a_{\min}$ is an alignment of $(\head(t_1),\ldots,\head(t_n))$ and $(\head(s_1),\ldots,\head(s_m))$ such that
\begin{itemize}
\item $\mathit{dec}(P,a_{\min})\in \mathit{B}(\mathcal{R},P)$,
\item for each $\mathit{dec}(P,a)\in \mathit{B}(\mathcal{R},P)$, let $D(a)$ be the derivation \[D(a)= \{ P\};\emptyset;\emptyset \Longrightarrow^{*}_{ \mathcal{G}  } \mathit{dec}(P,a); S'; \sigma' \Longrightarrow^{*}_{\mathcal{G}_{base}} \emptyset; S; \sigma_a.\] Then for each $a\neq a_{\min}$, the corresponding $D(a)$ computes $\sigma_a$ such that $X\sigma_a$ is more general than $X\sigma_{a_{\min}}$, where $\sigma_{a_{\min}}$ is computed by $D(a_{\min})$. If there are several such $a_{\min}$'s, $C_{(\mathcal{R},\mathcal{G})}(P)$ is defined as one of them (chosen by some heuristics).
\end{itemize}
\end{definition}

The choice function outlined above uses the linear time procedure ${\mathcal{G}_{base}}$ to make a choice between the various possible alignments. Notice that we use associative decomposition for $\{P \};\emptyset;\emptyset \Longrightarrow^{*} \mathit{dec}(P,a); S'; \sigma'$ and syntactic decomposition in the derivation $\mathit{dec}(P,a); S'; \sigma' \Longrightarrow^{*} \emptyset; S; \sigma_a$.

\begin{theorem}
\label{thm:AUPWeakAnti}
A $(\mathcal{R}_{A}^{k},C_{(\mathcal{R}_{A}^{k},\mathcal{G}_{A})},\mathcal{G}_{A})$-optimal higher-order $\{{ \sf A}\}$-refined pattern generalization for a {\em total  $k$-determined} AUP $X(\vect{x}):t\triangleq s$ can be computed in $O(n^2)$ where $n$ is the size of the AUP.
\end{theorem}
\begin{proof}

This follows from the existence of a linear algorithm for the computation of lggs using  $\mathcal{G}_{base}$ and the linear time algorithm of theorem~\ref{thm:AUPAnti}. Note that $k$ is constant and thus does not  show up in complexity statement. 
\end{proof}

\subsection{Step Optimal Generalization for Full Associativity}
\label{subsec:full}
Completely dropping the determinedness restrictions on the AUPs containing associative function symbols is the same as considering $O(n)$-determined AUPs. We have already shown that this problem is naively solvable by an exponential procedure, even when we consider $O(1)$-determined AUPs. In this section we again consider the problem of finding a $(\mathcal{R}_{A}^{^{O(n)}},C_{(\mathcal{R}_{A}^{O(n)},\mathcal{G}_{A})},\mathcal{G}_{A})$-optimal generalization where $n$ in the landau-notation refers to the maximum number of arguments of any subterms  in the given AUP.  However, this time the resulting algorithm is  cubic in complexity being that $r$ in r--determined is no longer a constant. By $\mathcal{R}_{A}^{^{O(n)}}$ we mean the singleton rigidity function which instead of looking for an $r$-determined subsequence just considers the largest feasible multiset intersection.

\begin{theorem}
\label{thm:AUPFull}
A $(\mathcal{R}_{A}^{^{O(n)}},C_{(\mathcal{R}_{A}^{O(n)},\mathcal{G}_{A})},\mathcal{G}_{A})$-optimal higher-order $\{{ \sf A}\}$-refined  pattern generalization for an AUP $X(\vect{x}):t\triangleq s$ can be computed in $O(n^3)$ time where $n$ is the size the AUP. 
\end{theorem}

Now that we have completed our analysis of associative function symbols, the simplest of the cases we consider, we move on to the more interesting cases of unit and commutative decomposition as well as the combinations of these algebraic properties.

\section{Commutative Case}
\label{sect:special:C}

Notice that in the case of commutative decomposition if all four terms (or three terms) have the same head symbol we end up with similar issues as in the associativity case. We can use strict $2$-determined to restrict the considered AUPs.

\begin{theorem}
\label{thm:AUPAntiCU}
A higher-order $\{C\}$-refined pattern generalization, for a {\em total strict 1-determined} AUP can be computed in linear time.
\end{theorem}
\begin{proof}
Similar to the proof of Theorem \ref{thm:AUPAnti}. 
\end{proof}

Note that the case $f(t_{1},t_{2}) \triangleq f(s_{1},s_{2})$,  where $\mathit{head}(t_1)=\mathit{head}(s_1)$ and $\mathit{head}(t_2)=\mathit{head}(s_2)$, is considered by the procedure of Theorem \ref{thm:AUPAntiCU}, but not $f(t_{1},t_{2}) \triangleq f(s_{2},s_{1})$ This is an issue with the definition of total strict 1-determined. We can fix this problem by performing an addition check to see if a permutation of the terms on the left or right side results in a better alignment. We now present a procedure for full commutativity, that is without restrictions which has a quadratic complexity (see Theorem \ref{thm:AUPWeakAnti}. 

\begin{definition}
Let $w_{1} =  (a_1,\ldots, a_n)$ and  $w_{2} = (b_1,\ldots, b_m)$ be sequences of symbols and $k\ge 1$ be an integer. We define the rigidity function $\mathcal{R}_{\sf C}$ returning all alignments.
\end{definition}

When the rigidity function $\mathcal{R}_{\sf C}$ is used all by our procedure there will be at most 4 alignments.
\begin{corollary}
\label{thm:AUPCU}
A  $(\mathcal{R}_{\sf C},C_{(\mathcal{R}_{\sf C},\mathcal{G}_{\left\lbrace {\sf C}\right\rbrace })},\mathcal{G}_{\left\lbrace {\sf C}\right\rbrace })$-optimal  higher-order $\{ C\}$-refined pattern generalization  for an AUP can be computed in quadratic time.
\end{corollary}

\section{Associative-Commutative Case}
\label{sect:special:AC}

In this section we consider functions $f$ such that $Ax(f) = \lbrace \sf A,C\rbrace$. Unfortunately, when a function is both associative and commutative, the number of possible decomposition paths is even greater than the previously considered cases and thus we need to further restrict the term structure. To provide a better understanding of why this is the case, consider a $k$-determined AUP where the multiset intersection is of size $O(k)$ on only contains one function symbol. This implies that there are $O(k^2)$ possible decompositions of the terms in the first multiset intersection of the terms containing $k$ alignments. This is not even considering that there might be more than one function symbol in the AUP. The problem is that the more terms with the same head symbol, the more combinations we must check. Unlike commutativity, which considers a fixed number of terms, and associativity, which enforces an ordering on terms, associative-commutativity allows an arbitrary number of arguments with no fixed ordering on the terms. We can get around this problem by considering special cases of AUPs  where arguments of an associative-commutativity symbol have distinct heads.


Unfortunately, the concept of (strict) $k$-determined AUPs does not lead to a linear algorithm in the case of AC-generalization.  Actually, this concept is not even meaningful for such an equational theory, since terms are not ordered in any particular way. Instead, we need to consider so called $(k,l)$-distinct AUPs, which are defined as follows:

\begin{definition}
Let $P\equiv X(\vect{x}):\lambda x_{1},\ldots, x_{l}.f(t_{1},\ldots,t_{n}) \triangleq \lambda y_{1},\ldots, y_{k}.$ $f(s_{1},\ldots,s_{m})$, $\pahm(f(t_{1},$ $\ldots,t_{n}),f(s_{1},\ldots,s_{m})) =  \left\langle T,S\right\rangle$, and $Ax(f)=\{\sf A,C \}$.  We say that $P$ is {\em $(k,l)$-distinct} if  each $h \in T\cap S$ occurs at most $k$ times in $w_{1}$ and at most $k$ times in $w_{2}$, the number of symbols in  $T \cap S \leq l$ and $T\setminus (T \cap S ) \emptyset$ iff  $S\setminus (T \cap S ) \emptyset$. We say $P\equiv X(\vect{x}): \lambda x_{1},\ldots, x_{w}.t\triangleq \lambda y_{1},\ldots, y_{r}.s$ is {\em total $(k,l)$-distinct} if $|t| = |s|=1$ or for every pair of subterms $(t',s')$ of $t$ and $s$ such that $\mathit{head}(t')=\mathit{head}(s')$, the AUP  $Y(\vect{y}): t'\triangleq s'$ is total $(k,l)$-distinct.
\end{definition}

This concept is much simpler than $k$-determined in that it basically splits the arguments of the left and right side of the given AUP  into at most $l$ sections dependent on the head symbols of the arguments. Also, for head function symbol, there should be at most $k$ occurrences of it and the result of decomposition can only be an empty term if the terms of both the left and right side of the AUP are empty. 

When an AUP is total $(1,l)$-distinct there is only one way to decompose the AUP, i.e. either a given symbol shows up in both $w_{1}$ and $w_{2}$ once and can be aligned, or it cannot be aligned and is generalized by a new variable.

This leads to the following results: 

\begin{theorem}
\label{ACOne}
A higher-order $\{{ \sf A,C}\}$-refined pattern generalization for a  {\em total $(1,l)$-distinct} AUP can be solved in linear time.
\end{theorem}
\begin{proof}
Similar to the proof of Theorem~\ref{thm:AUPAnti}.
\end{proof}

If we attempt to relax these constraints the time complexity of the algorithm increases substantially, even when we consider the case of $(2,l)$-distinct AUPs under our restricted optimality condition. 

\begin{definition}
Let $w_{1} =  (a_1,\ldots, a_n)$ and  $w_{2} = (b_1,\ldots, b_m)$ be sequences of symbols. We define the singleton rigidity function $\mathcal{R}_{AC}^{(k,l)}$ as follows 
\begin{equation}
\mathcal{R}_{AC}^{(k,l)}(w_1,w_2) = \left\lbrace 
\begin{array}{c|c} 
\left\lbrace a_{l}\left[ i,j\right]\ \middle\vert\ 
\begin{array}{c} a_{i} = b_{j} \\ 1\leq i \leq n\\ 1\leq j\leq m  \end{array}\right\rbrace  &  
\begin{array}{c}  \mbox{if }(w_{1},w_{2}) \mbox{ is } (k,l)\mbox{-distinct} \end{array}\\
 \emptyset &  \mbox{otherwise}
\end{array} \right. 
\end{equation}
\end{definition}

\begin{theorem}
\label{thm:AUPWeakAC}
A $(\mathcal{R}_{AC}^{(k,l)},C_{(\mathcal{R}_{AC}^{(k,l)},\mathcal{G}_{AC})},\mathcal{G}_{AC})$-optimal higher-order $\{{ \sf A,C}\}$-refined pattern generalization for a {\em total $(k,l)$-distinct} AUP can be computed in  $O(k^{2\cdot l}\cdot n^{2})$ time for $n$ as input size.
\end{theorem}
\begin{proof}
There are $O(k^{2})$ ways to pair the terms with the same head and there are $l$ blocks thus there are $O(k^{2\cdot l})$ computations using $\mathcal{G}_{base}$ (complexity $O(n)$) to be performed on an AUP with size $n$. 
\end{proof}
Obviously, computing the full set of $E$-refined generalizations from the results of Theorem~\ref{thm:AUPWeakAC} using a naive method would take in the order of $O(k^{2\cdot l\cdot n})$ time.

\section{Conclusion}
\label{sect:concl}
Higher-order equational anti-unification algorithm presented in this paper combines higher-order syntactic anti-unification rules with the decomposition rules for associative, commutative and associative-commutative function symbols. This gives a modular algorithm, which can be used for problems with different symbols from different theories without any adaptation. 

Higher order A-, C-, and AC-anti-unification problems are finitary. In practice, often it is desirable to compute only one answer, which is the best one with respect to some predefined criterion. We defined such an optimality criterion, which basically means that an optimal equational solution should be at least a good as the syntactic lgg. We then identified problem forms for which optimal solutions can be computed fast (in linear or polynomial time) by a greedy approach.

\bibliography{hoepau}

\end{document}

%% file: SpecialSymbols.tex
\usepackage[numbers]{natbib}
\usepackage{amsmath,amsfonts,amssymb}
\usepackage{listings}
\usepackage{verbatim}
\usepackage{alltt}
\usepackage{mathabx}
\usepackage{esvect}
\usepackage{stmaryrd}
\usepackage[normalem]{ulem}
\usepackage{paralist}
\usepackage[Symbol]{upgreek}
\usepackage{mathptmx}

\makeatletter
\def\infrule#1#2#3{\@ifnextchar[{\@infrule{#1}{#2}{#3}}{\@infrule{#1}{#2}{#3}[*]}}%

\def\@infrule#1#2#3[#4]{
 \par\bigbreak
 \vtop{
  \leavevmode\null
  \textsf{#1:}\kern.5em\textbf{#2}\\[\smallskipamount]
  \ifx*#4 
   \null\qquad$#3$
  \else 
   \setbox30=\hbox{\qquad$#3$\qquad#4}%
   \ifdim\wd30>\hsize
    \null\qquad$#3$\par\kern-\parskip\smallskip#4
   \else
    \null\qquad$#3$\qquad#4
   \fi
  \fi
 }%
} \makeatother

\allowdisplaybreaks
\newcommand{\moregen}{\mathrel{\leq\kern-2.9pt\raise0.8pt\hbox{\llap{$\cdot$}}}}

\newcommand{\ignore}[1]{}

 \newcommand{\cE}{\mathcal E}

 \newcommand{\dom}{\mathrm{Dom}}
 \newcommand{\ran}{\mathrm{Ran}}

 \newcommand{\pahs}{\mathit{pahs}}
 \newcommand{\pahm}{\mathit{pahm}}

 \newcommand{\depth}{\mathit{depth}}
 
 \newcommand{\Lra}{\Longrightarrow}
 
 \newcommand{\size}[1]{{|#1|}}
 
 \newcommand{\vars}{\mathrm{Vars}}

\newcommand{\head}{\mathit{head}}

\newcommand{\ids}{\emptyset}
\newcommand{\vect}[1]{\ensuremath{\vv{#1}}}

\newcommand*{\ldblbrace}{\{\mskip-5mu\{}
\newcommand*{\rdblbrace}{\}\mskip-5mu\}}

\newcommand{\sA}{{\sf A}}
\newcommand{\sC}{{\sf C}}